\documentclass[twocolumn,letterpaper,accepted=2020-03-27]{quantumarticle}

\pdfoutput=1
\usepackage[utf8]{inputenc}
\usepackage[usenames,dvipsnames]{xcolor}
\usepackage[english]{babel}
\usepackage[T1]{fontenc}
\usepackage{amsmath,amssymb,amsfonts,microtype,ragged2e}

\usepackage{appendix,hyphenat,mathtools,dsfont,bm,bbm,amsthm,tikz,pgfplots}
\usepackage[makeroom]{cancel}
\usepackage[numbers,sort&compress]{natbib}

\definecolor{mygrey}{gray}{0.35}
\definecolor{myblue}{rgb}{0.2,0.2,0.8}
\definecolor{myzard}{cmyk}{0,0,0.05,0}
\definecolor{mywhite}{rgb}{1,1,1}
\definecolor{myred}{rgb}{0.9,0.1,0.}
\definecolor{dgreen}{rgb}{0.0, 0.5, 0.0}
\usepackage[colorlinks=true,citecolor=myblue,linkcolor=myblue,urlcolor=myblue]{hyperref}
\usepackage[capitalize]{cleveref}
\usepackage{bookmark}

\newtheorem{theorem}{Theorem}
\newtheorem{lemma}[theorem]{Lemma}

\newtheorem{definition}[theorem]{Definition}
\newtheorem{corollary}[theorem]{Corollary}

\newcommand{\ket}[1]{\lvert#1\rangle} 

\newcommand{\ketbra}[2]{\ensuremath{\lvert#1\rangle\!\langle#2\rvert}}
\newcommand{\m}[1]{\mathsf{#1}}
\makeatletter
\def\mbig{\bBigg@{1.2}}
\def\mbigl#1{\mathopen{\bBigg@{1.2}#1}}
\def\mbigr#1{\mathclose{\bBigg@{1.2}#1}}
\renewcommand*\env@matrix[1][*\c@MaxMatrixCols c]{%
   \hskip -\arraycolsep%
   \let\@ifnextchar\new@ifnextchar%
   \array{#1}%
}
\makeatother

\newcommand*\bgcol[2]{{%
   \ifmmode%
      \mathchoice%
         {\colorbox{#1}{$\displaystyle#2$}}%
         {\colorbox{#1}{$\textstyle#2$}}%
         {\colorbox{#1}{$\scriptstyle#2$}}%
         {\colorbox{#1}{$\scriptscriptstyle#2$}}%
   \else%
      \colorbox{#1}{#2}%
   \fi%
}}

\def\one{\sigma_0}
\def\C{{\ensuremath{\mathbb C}}}

\def\N{{\ensuremath{\mathbb N}}}

\DeclareMathOperator\tr{Tr}
\def\ii{\mathrm i}
\def\ee{\mathrm e}
\def\d{\mathop{}\!\mathrm d\mathchoice{}{}{\kern-.09em}{\kern-.09em}}
\newcommand*\D[2][]{\frac{\d #1}{\d #2}}
\def\peq{\mathrel{\phantom{=}}}

\newcommand*\mat[1]{\begin{pmatrix} #1 \end{pmatrix}}

\def\sx{\sigma_1}
\def\sy{\sigma_2}
\def\sz{\sigma_3}

\def\rhoS{\rho_{\mathrm s}}
\def\gen{\mathcal L}
\def\map{\mathcal E}

\def\bound{\mathcal B}
\def\hilb{\mathcal H}
\def\order{\mathcal O}

\makeatletter
\def\PP{_{\mathrm{p p}\vphantom j}}
\def\PC{_{\mathrm{p c}\vphantom j}}
\def\CP{_{\mathrm{c p}\vphantom j}}
\def\CC{_{\mathrm{c c}\vphantom j}}
\def\cL{\mathcal L}
\makeatother
\def\compose{\mathbin\circ}

\let\Re\relax
\DeclareMathOperator\Re{Re}
\let\Im\relax
\DeclareMathOperator\Im{Im}

\pgfplotsset{compat=1.14}

\begin{document}
   \title{Accessible coherence in open quantum system dynamics}

   \author{Mar{\'i}a Garc{\'i}a D{\'i}az}
   \affiliation{F{\'i}sica Te\`{o}rica: Informaci{\'o} i Fen\`{o}mens Qu\`{a}ntics, Departament de F{\'i}sica, Universitat Aut\`{o}noma de Barcelona, 08193 Bellaterra (Barcelona), Spain}
   \orcid{0000-0002-4175-4090}
   \listcsgadd{author1affiliations}{\ref{fn:author}}
   
   \author{Benjamin Desef}
   \affiliation{Institute  of  Theoretical  Physics and IQST,  Universit{\"a}t  Ulm,  Albert-Einstein-Allee  11, D-89069  Ulm,  Germany}
   \orcid{0000-0003-2083-7820}
   \listcsgadd{author2affiliations}{\ref{fn:author}}
   
   \author{Matteo Rosati}
   \affiliation{F{\'i}sica Te\`{o}rica: Informaci{\'o} i Fen\`{o}mens Qu\`{a}ntics, Departament de F{\'i}sica, Universitat Aut\`{o}noma de Barcelona, 08193 Bellaterra (Barcelona), Spain}
   \orcid{0000-0002-8972-2936}
   
   \author{Dario Egloff}
   \affiliation{Institute  of  Theoretical  Physics and IQST,  Universit{\"a}t  Ulm,  Albert-Einstein-Allee  11, D-89069  Ulm,  Germany}
   \affiliation{Institute  of  Theoretical  Physics,  Technical University Dresden, D-01062 Dresden, Germany}
   \orcid{0000-0002-7874-0258}
   
   \author{John Calsamiglia}
   \affiliation{F{\'i}sica Te\`{o}rica: Informaci{\'o} i Fen\`{o}mens Qu\`{a}ntics, Departament de F{\'i}sica, Universitat Aut\`{o}noma de Barcelona, 08193 Bellaterra (Barcelona), Spain}
   \orcid{0000-0003-1735-1360}
   
   \author{Andrea Smirne}
   \affiliation{Institute  of  Theoretical  Physics and IQST,  Universit{\"a}t  Ulm,  Albert-Einstein-Allee  11, D-89069  Ulm,  Germany}
   \affiliation{Dipartimento di Fisica Aldo Pontremoli, Universit\`{a} degli Studi di Milano, via Celoria 16, 20133 Milan, Italy}
   \orcid{0000-0003-4698-9304}
   
   \author{Michalis Skotiniotis}
   \affiliation{F{\'i}sica Te\`{o}rica: Informaci{\'o} i Fen\`{o}mens Qu\`{a}ntics, Departament de F{\'i}sica, Universitat Aut\`{o}noma de Barcelona, 08193 Bellaterra (Barcelona), Spain}
   \orcid{0000-0001-6935-7460}
   
   \author{Susana F. Huelga}
   \affiliation{Institute  of  Theoretical  Physics and IQST,  Universit{\"a}t  Ulm,  Albert-Einstein-Allee  11, D-89069  Ulm,  Germany}
   \orcid{0000-0003-1277-8154}
   
   \begin{abstract}
      Quantum coherence generated in a physical process can only be cast as a potentially useful resource if its effects can be detected at a later time.
      Recently, the notion of non\hyp coherence\hyp generating\hyp and\hyp detecting (NCGD) dynamics has been introduced and related to the classicality of the statistics associated with sequential measurements at different times.
      However, in order for a dynamics to be NCGD, its propagators need to satisfy a given set of conditions for \emph{all} triples of consecutive times.
      We reduce this to a finite set of $d(d-1)$ conditions, where $d$ is the dimension of the quantum system, provided that the generator is time\hyp independent.
      Further conditions are derived for the more general time\hyp dependent case.
      The application of this result to the case of a qubit dynamics allows us to elucidate which kind of noise gives rise to non\hyp coherence\hyp generation\hyp and\hyp detection.
   \end{abstract}

   \maketitle
   
   {\def\thefootnote{*}
    \footnotetext{\label{fn:author}These authors contributed equally.}}

   \section{Introduction}\label{sec:introduction}
      Much experimental effort in quantum physics focuses on the creation, maintenance, and subsequent detection of coherent superpositions of quantum states~\cite{Ronnow2014, Campbell2017, Preskill2018}, a distinctive feature of the quantum formalism that furnishes significant advantage in many communication~\cite{Gisin2007}, computation~\cite{Galindo2002, Georgescu2014}, and metrological tasks~\cite{Giovannetti2011,Toth2014,Degen2017}.
      Recently, resource theories for coherent superpositions have been developed in order to provide a cohesive and quantitative description of quantum coherence, both for states~\cite{Aberg2006, Baumgratz2014, Streltsov2017, Winter2016, Egloff2018} as well as operations~\cite{Bendana2017, Garciadiaz2018, Theurer2019, Liu2019, Gour2019}.
      
      Assessing the coherence capabilities of quantum dynamical maps is a subtle task.
      For example, the mere ability of a quantum dynamics to generate or detect coherence is of no practical advantage unless this coherence can be harnessed in a beneficial way for some task.
      In this sense, a prerequisite for a quantum dynamical evolution to generate \emph{resourceful} coherence for a given task is that the coherence it generates can be detected in terms of discriminable statistics of subsequent measurement outcomes associated with this task.
      In addition, in \cite{Smirne2017} it has been shown that, under proper conditions, dynamics which are non\hyp coherence\hyp generating\hyp and\hyp detecting (NCGD) are strictly related to the classicality of the statistics associated with sequential measurements.
      
      In this work we consider an open\hyp system dynamics and propose definite criteria to assess whether it is able to generate and detect coherences.
      Our main result shows that for dynamical maps stemming from time\hyp independent generators, the---in principle countless---conditions of~\cite{Smirne2017} reduce to a \emph{finite} set of necessary \emph{and} sufficient conditions.
      In addition, those are given in terms of the \emph{generator} of the quantum dynamical evolution.
      We then extend the notion of NCGD to time\hyp dependent dynamics.
      In this case, we can also derive relations between NCGD and the generator, though now in general we need uncountably infinitely many conditions.
      
      The article is structured as follows: We review the necessary background on open\hyp system dynamics (Sec.~\ref{oqs}) as well as the concept of NCGD (Sec.~\ref{ncgd_subsection}) and its connection with the statistics of sequential measurements.
      Sec.~\ref{main_section} contains our main results elaborating on when dynamical evolutions are NCGD based on their generators, and in Sec.~\ref{qubit_ncgd} we exemplify our criteria in the context of a Ramsey protocol widely used in precision spectroscopy.
      For ease of exposition, we defer the proofs of all theorems and propositions to the appendices.
     
   \section{Background}\label{background}
      \subsection{Open quantum systems}\label{oqs}
         A realistic description of a quantum system must take into account that every system is open, i.e., it interacts with the surrounding environment~\cite{Breuer2002}.
         In many circumstances, it is possible to provide such a description by a time\hyp local quantum master equation~(QME):
         \begin{equation}\label{eq:gener}
            \D{t} \rhoS(t) = \gen(t)[\rhoS(t)]\text,
         \end{equation}
         where $\gen(t)$ is the \emph{dynamical generator} of the evolution and $\rhoS$ the reduced state of the system.
         Any $\gen(t)$ that is both trace- and hermiticity\hyp preserving can be uniquely decomposed as~\cite{Gorini1976}
         \begin{equation}
            \begin{aligned}
               \gen(t)[\rhoS]
               & = -\ii[H(t), \rhoS] \\
               & \peq + \sum_{i,\kern.4pt j=1}^{d^2 -1} D_{ij}(t) \bigl(
                      F_i^{\vphantom\dagger} \rhoS^{\vphantom\dagger} F_j^\dagger -
                      \frac{1}{2}\{F_j^\dagger F_i^{\vphantom\dagger},
                                   \rhoS^{\vphantom\dagger}\}
                   \bigr)\text,
            \end{aligned}
            \label{eq:Lindblad_gen}
         \end{equation}
         where $d < \infty$ is the dimension of the Hilbert space of the system $\hilb$, $H(t)$ a Hermitian operator, $D(t)$ a Hermitian matrix, and $\{F_i^{}\}_{i=1}^{d^2}$ is an orthonormal operator basis with $F_{d^2} = \mathbbm{1}/\sqrt{d}$ and $\tr(F_i^\dagger F_j^{}) = \delta_{ij}$.
         
         Upon integration, the QME leads to a family of trace\hyp preserving~(TP) propagators $\map_{t_2, t_1}$ satisfying $\rhoS(t_2) = \map_{t_2, t_1}[ \rhoS(t_1) ]\ \forall t_2 \geq t_1 \geq 0$.
         
         In particular, we will first address dynamics for which $H$ and $D$ are time\hyp independent.
         This class gives rise to semigroups of the form $\map_t = \ee^{t \gen}$, where $\gen(t) \equiv \gen$ is time\hyp independent.
         Its most prominent representative is the Gorini\hyp Kossakowski\hyp Sudarshan--Lindblad (GKSL) form \cite{Gorini1976,Lindblad1976}, additionally demanding $D \geq 0$, which enforces complete positivity (CP) on all propagators.
         Despite the fact that a derivation of this form via a microscopic model involves several approximations \cite{Breuer2002}---which is by no means the only route to a GKSL QME \cite{Englert2002,Rivas2012,Breuer2002,Gorini1978,Kirsanskas2018}---it has shown remarkable success and applicability, in particular in the fields of quantum optics \cite{Walls2008} and in relevant problems of interest in solid\hyp state physics, like non\hyp equilibrium transport of charge and energy \cite{Brandes2005}.
         
         For even more generality, we will however also consider the case of a time\hyp dependent generator.
         
         Finally we shall refer to \emph{rank-$k$ noise} as those dynamical generators $\gen$ in Eq.~\eqref{eq:Lindblad_gen} for which $D$ has at most $k$ non-zero eigenvalues.
         
      \subsection{Non-coherence-generating-and-detecting dynamics}\label{ncgd_subsection}
         The realization that quantum coherence underpins the performance of many quantum information and communication tasks elevates it to a \emph{bona fide} resource.
         Over the past decade, a great deal of effort has gone into developing resource theories of coherence in an attempt to quantify and better leverage its use, both for states~\cite{Aberg2006, Baumgratz2014, Streltsov2017, Winter2016, Egloff2018} as well as operations~\cite{Bendana2017, Garciadiaz2018, Theurer2019, Liu2019, Gour2019}.
         Within these frameworks, the set of free states, $\mathcal I$, consists of all states that are diagonal in some fixed basis $\{ \ket{i} : i = 1, \dotsc, \dim(\hilb) \}$, while free operations are those that do not generate coherence out of incoherent states.
         A plethora of non\hyp coherence\hyp generating operations has been proposed~\cite{Chitambar2016a}, among which maximally incoherent operations (MIOs), defined as all CPTP maps $\mathcal M\colon \bound(\hilb) \to \bound(\hilb)$ such that $\mathcal M(\mathcal I) \subset \mathcal I$, constitute the largest class~\cite{Aberg2006}.
         Here, $\bound(\mathord\cdot)$ denotes the set of bounded operators.
         
         From the above definition of free operations it would seem that the only desideratum for a resourceful quantum operation is its ability to generate coherence.
         However, coherence in itself is of no value unless we are able to subsequently harness its presence in a beneficial manner~\cite{Theurer2019}.
         In order to do so, a dynamical map $\map_{t_1, 0}$ that generates coherence at some time $t_1$ must, at some later time $t_2> t_1$, be able to detect coherence (note that sometimes the word `activate'~\cite{liu2017} is used as a synonym for `detect').
         Therefore, what we are interested in are the coherence\hyp generating\hyp and\hyp detecting properties of the propagators $\{ \map_{t_2, t_1} : t_2 \geq t_1 \geq 0 \}$ associated with the dynamics.
         To that end we define non\hyp coherence\hyp generating\hyp and\hyp detecting (NCGD) dynamics as follows:
         
         \begin{definition}\cite{Smirne2017}
            A dynamics with propagator $\map_{t_2, t_1}$, $t_2 \geq t_1\geq 0$, is NCGD iff the condition
            \begin{equation}\label{ncgd_def}
               \Delta \compose \map_{t_3, t_2} \compose \Delta \compose \map_{t_2, t_1} \compose \Delta
               = \Delta \compose \map_{t_3, t_1} \compose \Delta
            \end{equation} 
            holds for all times $t_3 \geq t_2 \geq t_1 \geq 0$, where $\compose$ denotes composition of maps and $\Delta = \sum_{i = 1}^d \ketbra ii \cdot \ketbra ii$ is the complete dephasing map in the incoherent basis $\{ \ket i \}_{i = 1}^d$. \\
            Otherwise the dynamics is coherence\hyp generating\hyp and\hyp detecting (CGD).
            \label{def:NCGD}
         \end{definition}
        
        By looking at the statistics of measurement outcomes in the incoherent basis, it can be verified that CGD evolutions produce accessible coherence, i.e., coherence that affects the statistics of sequential measurements.
        In fact, under the assumptions that the dynamics of a quantum system is given by a CP semigroup and that the Quantum Regression Theorem holds~\cite{Lax1968, Swain1981, Carmichael1993, Breuer2002, Guarnieri2014}, the joint probability distribution arising from sequentially measuring a non\hyp degenerate observable of the quantum system is compatible with a classical stochastic process if and only if the dynamics is NCGD~\cite{Smirne2017}.
        The creation and subsequent detection of quantum coherence is thus, in this case, the distinctive feature of any genuine non\hyp classicality in the process.
        Note that here by classical process we mean any process whose statistics satisfies the Kolmogorov consistency conditions~\cite{Breuer2002, Feller1971}, motivated by the fact that---at least in principle---classical physics allows for noninvasive measurability.
        Violations of Leggett--Garg\hyp type inequalities \cite{Leggett1985,Emary2014,Huelga1995} can only be observed if a dynamics is CGD \cite{Smirne2017}, giving additional justification to this identification.
        Importantly, the notion of (N)CGD dynamics provides a theoretical framework that is amenable to deriving (quantitative) experimental benchmarks of coherence and its connection with non-classicality, as exemplified by the assessment of the time-multiplexed optical quantum walk in \cite{Smirne2019}.
        
        The direct connection between the dynamics of quantum coherence and non\hyp classicality can be extended beyond the case of CP semigroups, but does not hold for general evolutions~\cite{Strasberg2019, Milz2019}.
        
        Finally, we stress that the class of NCGD dynamical maps is not closed under composition: indeed, it is very easy to obtain CGD dynamics by composing two NCGD dynamical maps; the first being coherence\hyp generating but not detecting while the second one being coherence\hyp detecting but not generating.

   \section{Characterizing NCGD dynamics}\label{main_section}
      Given a quantum dynamics, Definition~\ref{def:NCGD} can be used to assess whether detectable coherences are generated.
      In principle, this would require performing three separate map tomography protocols for \emph{any} three given instants of time $t_1$, $t_2$, and $t_3$, and reconstruct the propagators involved in Eq.~\eqref{ncgd_def}, which would involve an uncountably infinite number of measurements.
      
      In this section we provide a finite set of necessary and sufficient conditions certifying the CGD properties of time\hyp independent dynamical evolutions.
      Importantly, these conditions pertain directly to the generator of the dynamics in Eq.~\eqref{eq:gener}.
      As the latter is the central tool to yield an explicit description, microscopically or phenomenologically motivated, of the open\hyp system dynamics of concrete physical settings~\cite{Breuer2002}, our result is relevant to certify the use of coherence in an open quantum system dynamics.
      For time\hyp dependent generators, we provide an infinite number of necessary and sufficient conditions on the generator.
      
      Recall that $\gen(t)\colon \bound(\hilb) \to \bound(\hilb)$ for all times $t$.
      Having fixed the complete dephasing superoperator $\Delta$, we can decompose $\bound(\hilb)$ into two orthogonal subspaces
      \begin{subequations}
         \begin{gather}
            \bound(\hilb)
              = \bound_{\mathrm p}(\hilb) \oplus
                \bound_{\mathrm c}(\hilb)\text,
            \shortintertext{where}
            \begin{aligned}
               \bound_{\mathrm p}(\hilb)
               & = \operatorname{Image}(\Delta)\text, \\
               \bound_{\mathrm c}(\hilb)
               & = \operatorname{Kernel}(\Delta)\text,
            \end{aligned}
         \end{gather}
         are the subspaces associated with the population and coherence basis elements respectively.
         In this basis, the matrix representation of the generator $\gen(t)$ is given by
         \begin{equation}\label{eq:gent}
            \gen(t)
              = \mat{\gen\PP(t) & \gen\PC(t) \\
                     \gen\CP(t) & \gen\CC(t)}\!\text,
         \end{equation}
         where, for example, $\gen\PC\colon \bound_{\mathrm c}(\hilb)\to\bound_{\mathrm p}(\hilb)$.
         \end{subequations}
      
      We are now ready to formulate our first result, which is a complete characterization of NCGD based solely on its time\hyp independent generator.
      \begin{theorem}
         For any\label{th:semi} time\hyp independent generator $\gen$ of a quantum dynamics it holds that
         \begin{equation*}
            \text{NCGD}
            \Leftrightarrow \mbigl(
               \gen\PC^{\vphantom j} \gen\CC^{j}\gen\CP^{\vphantom j} = 0
               \ \forall j \in \{ 0, \dotsc, d^2 - d -1 \}
            \mbigr)\text,
         \end{equation*}
         where $d = \dim\hilb$.
      \end{theorem}
      The proof in Appendix~\ref{app:proof} makes use of the decomposition in Eq.~\eqref{eq:gent}.
      It then employs the Cayley\hyp Hamilton theorem~\cite{Segercrantz1992} to reduce the uncountably infinitely many conditions in Eq.~\eqref{ncgd_def} to $d (d-1)$ conditions on the generator.
      
      As a direct generalization of Theorem~\ref{th:semi}, we can state the following Theorem for the more general case of a time\hyp dependent generator, under the regularity condition that the latter is analytic for all times considered.
      Its proof can be found in Appendix~\ref{app:eq}.
      \begin{theorem}
         For a\label{ncgd_eq} sufficiently regular $\gen$(t), we have
         \begin{multline*}
            \text{NCGD}
            \Leftrightarrow \forall t_n \geq \dotsb \geq t_1 \geq 0 \text,
                \ \forall n\geq 2 : \\
              \gen\PC(t_n) \compose \gen\CC(t_{n-1})\compose \dotsb \compose
              \gen\CC(t_2) \compose \gen\CP(t_1) = 0 \text.
         \end{multline*}
      \end{theorem}
      
      Even though this characterization of NCGD, in contrast to the previous case, now consists of infinitely many conditions, in certain special cases some simpler conditions already guarantee that a dynamics is NCGD.
      This is captured in the following Corollary.
      \begin{corollary}
         The\label{ncgd_suff} following conditions on the generator $\gen(t)$ of a quantum dynamics individually imply NCGD.
         \begin{itemize}
            \item $\gen\PC(t) = 0\ \forall t \geq 0$
            \item $\gen\CP(t) = 0\ \forall t \geq 0$
            \item $\gen(t_2) \compose \Delta \compose \gen(t_1)
                   = \gen(t_2) \compose \gen(t_1)\ \forall t_2 \geq t_1 \geq 0$
            \item $\bigl[\gen(t), \Delta\bigr] = 0\ \forall t \geq 0$.
         \end{itemize}
      \end{corollary}
      
      It is sufficient to find a single set of times for which the equality in Theorem~\ref{ncgd_eq} does not hold in order to guarantee that the dynamics is CGD.
      The simplest instance of this is captured in the following Corollary.
      \begin{corollary}
         For\label{ncgd_necc} a sufficiently regular $\gen(t)$, we have
         \begin{equation*}
            \text{NCGD}
            \Rightarrow \gen\PC(t_2) \compose \gen\CP(t_1) = 0 \ \forall t_2 \geq t_1 \geq 0\text.
         \end{equation*}
      \end{corollary}
   
   \section{(N)CGD dynamics for qubits}\label{qubit_ncgd}
      In this section, we will apply Theorem~\ref{th:semi} to the special case of a GKSL qubit dynamics.
      This will allow us to explicitly give the structure of an NCGD dynamics.
      
      Eq.~\eqref{eq:Lindblad_gen} in the normalized Pauli operator basis $\{\sigma_i : i = 0, \dotsc, 3 \}$ can be easily rewritten as
      \begin{equation}
         \cL[\rhoS]
         = \frac{1}{2} \sum_{i,\kern.4pt j = 0}^3 \m L_{i j}
              \bigl( [\sigma_i\rhoS, \sigma_j] + [\sigma_i, \rhoS\sigma_j] \bigr)\text,
         \label{eq:qbit_Lindblad}
      \end{equation}
      where $\m L\in\C^{4\times 4}$ is a Hermitian matrix.
      We will choose $(\one, \sz)$ as our incoherent basis and $(\sx, \sy)$ as the coherent one.
      With this choice, the matrix representation of Eq.~\eqref{eq:qbit_Lindblad} in the basis of Eq.~\eqref{eq:gent} is explicitly given by
      \begin{equation}\label{eq:split}
         \begin{split}
            \gen\PP
            & = -\mat{0              & 0 \\
                      2 \Im\m L_{12} & \m L_{11} + \m L_{22}} \\
            \gen\PC
            & = \mat{0                           & 0 \\
                     \Re\m L_{13} - \Im\m L_{02} & \Re\m L_{23} + \Im\m L_{01}} \\
            \gen\CP
            & = \mat{-2\Im\m L_{23} & \Re\m L_{13} + \Im\m L_{02} \\
                     2\Im\m L_{13}  & \Re\m L_{23} - \Im\m L_{01}} \\
            \gen\CC
            & = -\mat{\m L_{22} + \m L_{33}        & \Im\m L_{03} - \Re\m L_{12} \\
                      -\Re\m L_{12} - \Im\m L_{03} & \m L_{11} + \m L_{33}}\!\text.
         \end{split}
      \end{equation}
      
      Theorem~\ref{th:semi} states that NCGD is equivalent to
      \begin{equation}
         \gen\PC \gen\CP
         = \gen\PC \gen\CC \gen\CP
         = 0\text.\label{eq:NCGD_qubit}
      \end{equation}
      In particular, the dynamics is coherence non\hyp activating, i.e., $\gen\PC = 0$, when
      \begin{equation}
         \Re\m L_{13} = \Im\m L_{02}
         \quad\land\quad
         \Re\m L_{23} = -\Im\m L_{01},
         \label{eq:non-activating}
      \end{equation}
      while it is coherence non\hyp generating, i.e., $\gen\CP = 0$, when
      \begin{equation}
         \begin{split}
            \Re\m L_{13} = -\Im\m L_{02}
            & \quad\land\quad
            \Re\m L_{23} = \Im\m L_{01} \\
            & \quad\land\quad
            \Im\m L_{13} = \Im\m L_{23} = 0\text.
         \end{split}\label{eq:non-generating}
      \end{equation}
      Observe that both coherence non\hyp activating and coherence non\hyp generating dynamics can arise from the simplest open\hyp systems dynamics, namely rank-one Pauli noise.
      For example, assuming that all contributions $\m L_{0i}$ arise solely from the Hamiltonian of the system, the following rank-one dissipators $\bar{\m L} \in \C^{3 \times 3}$,
      \begin{equation}
         \begin{aligned}
            \bar{\m L}_{\text{non-act.}}
            & = \bm{r} \bm{r}^\top\text,
            & \bm{r}
            & = \mat{\, \Im\m L_{02} & -\Im\m L_{01} & 1}^{\!\top}\!\text; \\
            \bar{\m L}_{\text{non-gen.}}
            & = \bm{s} \bm{s}^\top\text,
            & \bm{s}
            & = \mat{-\Im\m L_{02} & \Im\m L_{01} & 1}^{\!\top}\!\text,
         \end{aligned}
      \end{equation}
      give rise to coherence non\hyp activating and coherence non\hyp generating dynamics respectively.
      
      Note, however, that one can have dynamical evolutions that are capable of both generating and detecting coherence, and yet are still NCGD.
      This occurs whenever coherence is generated in an orthogonal subspace to the one where it is detected.
      In the case of qubits this happens precisely when (assuming for simplicity that the denominators involved are different from 0)
      \begin{equation}
         \frac{\Im\m L_{13}}{\Im\m L_{23}}
         = \frac{\Re\m L_{13} - \Im\m L_{02}}
                {\Re\m L_{23} + \Im\m L_{01}}
         = \frac{\Im\m L_{01} - \Re\m L_{23}}
                {\Re\m L_{13} + \Im\m L_{02}}
         \label{eq:NCGD_gen1}
      \end{equation}
      and
      \begin{equation}
         \begin{split}
            \m L_{11} - \m L_{22}
            & = \frac{(\Im\m L_{02} - \Re\m L_{13})(\Im\m L_{03} - \Re\m L_{12})}
                     {\Im\m L_{01} + \Re\m L_{23}} + {} \\
            & \peq \frac{(\Im\m L_{01} + \Re\m L_{23})(\Im\m L_{03} + \Re\m L_{12})}
                        {\Im\m L_{02}-\Re\m L_{13}}\text.
         \end{split}
         \label{eq:NCGD_gen3}
      \end{equation}
      Eq.~\eqref{eq:NCGD_gen1} is equivalent to the first condition in Eq.~\eqref{eq:NCGD_qubit}, $\gen\PC \gen\CP = 0$; the precise relationship among several coefficients of the dynamical map ensures that coherence is generated in a subspace orthogonal to that of coherence detection.
      Likewise, Eq.~\eqref{eq:NCGD_gen3} rules out the second\hyp order coupling, $\gen\PC \gen\CC \gen\CP = 0$.
      
      Let us illustrate our findings with a concrete, and prac\-ti\-cal\-ly relevant physical example; the Ramsey scheme deployed in interferometry, spectroscopy and atomic clocks.
      The simplest, non\hyp trivial case of such a scheme is that of rank-one Pauli noise in the same direction as the Hamiltonian evolution---assumed without loss of generality to be $H = \sz$---whose dynamics is given by
      \begin{equation}
         \gen[\rhoS]
         = -\ii\omega[\sz, \rhoS] + \gamma(\sz\rhoS\sz - \rhoS/2)\text,
         \label{eq:atomic_clock}
      \end{equation}
      where $\omega$ is the detuning from the reference field.
      Note that due to the normalization of the Pauli matrices, $\sz^2 = \frac{\mathbbm1}{2}$.
      In the Ramsey scheme, the atoms---approximated as qubits---are first prepared in eigenstates of $\sx$, then subjected to the evolution generated by Eq.~\eqref{eq:atomic_clock}, and subsequently measured in the eigenbasis of $\sx$.
      Choosing $\bound_{\mathrm p}(\hilb) = ( \one, \sx )$ as our incoherent basis, $\bound_{\mathrm c}(\hilb) = ( \sy, \sz )$, and using the matrix representation introduced in Eq.~\eqref{eq:gent}, the generator of Eq.~\eqref{eq:atomic_clock} can be written as
      \begin{equation}
         \gen = \mat{
            0 &             0 &             0 & 0 \\
            0 &       -\gamma & -\sqrt2\omega & 0 \\
            0 &  \sqrt2\omega &       -\gamma & 0 \\
            0 &             0 &             0 & 0
         }\!\text.
         \label{eq:genatomic}
      \end{equation}
      We can assess the CGD properties of such a setup by looking at the distance between the left- and right-hand sides of Eq.~\eqref{ncgd_def}, as 
      measured via the trace distance.
      Defining
      \begin{equation}
         \begin{split}
            p_{\pm}(t_3)
            & = \tr\bigl(\ketbra{\pm}{\pm}\, \map_{t_3, t_1} \compose \Delta[\rhoS] \bigr) \\
            q_{\pm}(t_3, t_2)
            & = \tr\bigl(\ketbra{\pm}{\pm}\, \map_{t_3, t_2} \compose \Delta \compose \map_{t_2, t_1} \compose \Delta[\rhoS] \bigr)\text,
         \end{split}
         \label{eq:prob_dists}
      \end{equation}
      Figure~\ref{fig:atomic_clock} shows the trace distance $\max_{\rhoS \in \mathcal I} \lVert\bm{p}(t_3) - \bm{q}(t_3, t_2)\rVert$ as a function of the intermediate time $t_2$ for various values of the ratio $\gamma/\omega$.
      The presence of coherence in the dynamics is most prominent half-way through the evolution
      and, indeed, it is suppressed by a stronger rate $\gamma$.
      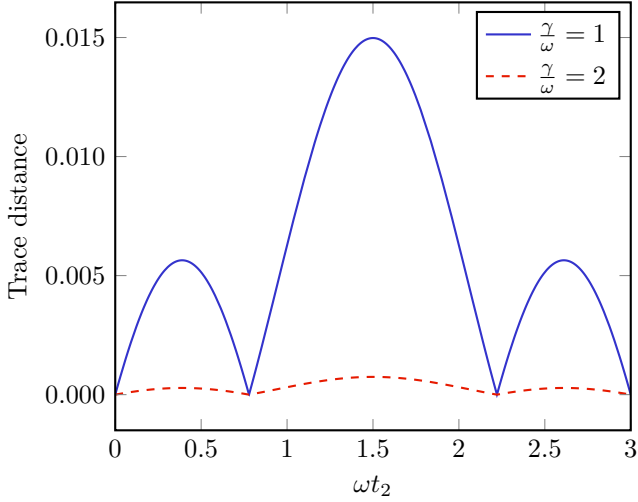
\begin{figure}
         \begin{tikzpicture}
            \begin{axis}[xlabel=$\omega t_2$, ylabel=Trace distance,
                         legend style={row sep=1mm},
                         enlarge x limits=false, scaled ticks=false,
                         yticklabel style={
                            /pgf/number format/.cd,
                            fixed, fixed zerofill, precision=3
                         }, axis on top, thick]
               \addplot[no marks, myblue] table {fig1-2.dat};
               \addlegendentry{$\frac\gamma\omega = 1$}
               \addplot[no marks, myred, dashed] table {fig1-1.dat};
               \addlegendentry{$\frac\gamma\omega = 2$}
            \end{axis}
         \end{tikzpicture}
         \caption{%
            Coherence generated and detected, as measured via the trace distance between the probability distributions of Eq.~\eqref{eq:prob_dists}, maximized over $\rhoS$, for the open\hyp system evolution described by Eq.~\eqref{eq:atomic_clock}.
            The total evolution time is fixed to $\omega t_3 = 3$ and the trace distance is plotted as a function of the intermediate time $0 \leq t_2 \leq t_3$.%
         }
         \label{fig:atomic_clock}
      \end{figure}
      
      Let us now investigate a complementary scenario in which we also include components orthogonal to the Hamiltonian in the noise.
      Specifically, consider the open\hyp system dynamics
      \begin{equation}
         \gen[\rhoS]
         = -\ii \omega [\sz, \rhoS] + \frac{1}{2}
            \sum_{i, j = 1}^{3} \gamma_{i j}\bigl([\sigma_i\rhoS, \sigma_j] + [\sigma_i, \rhoS\sigma_j]\bigr)\text,
         \label{eq:atomic_ramsey}
      \end{equation}
      where $\gamma_{i j} = \gamma_{j i}^*$ are the damping rates; still, our incoherent basis is $(\one, \sx)$.
      
      Coherence non\hyp generating dynamics corresponds to $\gamma_{1 2} = -\sqrt2 \omega$ and $\gamma_{1 3} = 0$, whereas coherence non\hyp activating dynamics is given by $\Re\gamma_{1 2} = \sqrt2 \omega$ and $\Re\gamma_{1 3} = 0$.
      
      To investigate the more general notion of NCGD dynamics, we look at the matrix representation for the corresponding generator; assuming for the sake of simplicity that
      all $\gamma_{i j}$ be real, it reduces to
      \begin{equation}
         \gen = \mat{
            0 & 0 & 0 & 0 \\
            0 & -\gamma_{2 2} - \gamma_{3 3} & -\sqrt2 \omega + \gamma_{1 2} & \gamma_{1 3} \\
            0 & \sqrt2 \omega + \gamma_{1 2} & -\gamma_{1 1} - \gamma_{3 3} & \gamma_{2 3} \\
            0 & \gamma_{1 3} & \gamma_{2 3} & -\gamma_{1 1} - \gamma_{2 2}
         }\text.
      \end{equation}
      It can be verified that
      \begin{equation}
         \begin{split}
            \gen\PC \gen\CP
            & = \gen\PC \gen\CC \gen\CP = 0 \\
            \Leftrightarrow 2\omega^2
            & = \gamma_{1 2}^2 + \gamma_{1 3}^2 \\
            {}\land \gamma_{1 3}^2 (\gamma_{2 2} - \gamma_{3 3})
            & = 2 \gamma_{1 2} \gamma_{1 3} \gamma_{2 3}
         \end{split}
      \end{equation}
      so that indeed, the CGD capabilities of the dynamical evolution depend on the damping rates $\gamma_{1 2}$ and $\gamma_{1 3}$ that mix coherent with incoherent components.
      
      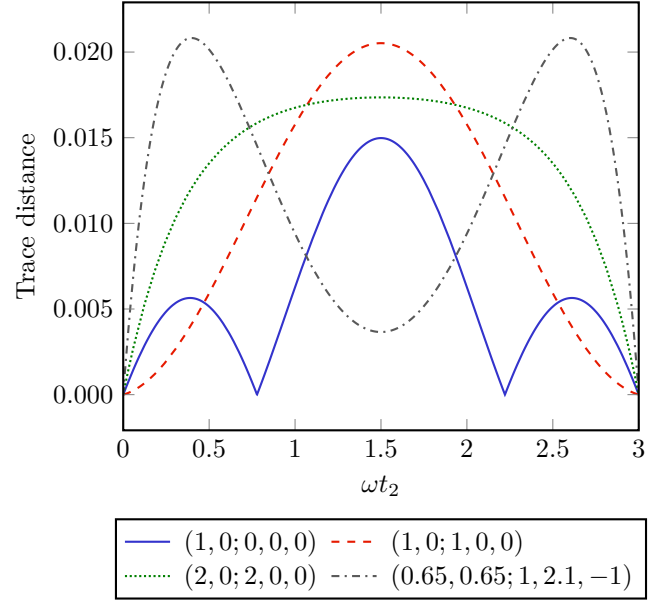
\begin{figure}[t]
         \begin{tikzpicture}
            \begin{axis}[xlabel=$\omega t_2$, ylabel=Trace distance,
                         legend style={at={(xticklabel cs:.5)}, anchor=north,
                                       column sep=1mm, yshift=-7mm},
                         legend columns=2, legend cell align=left,
                         enlarge x limits=false, scaled ticks=false,
                         yticklabel style={
                            /pgf/number format/.cd,
                            fixed, fixed zerofill, precision=3
                         }, axis on top, thick]
               \addplot[no marks, myblue] table {fig2-1.dat};
               \addlegendentry{$(1, 0; 0, 0, 0)$}
               \addplot[no marks, myred, dashed] table {fig2-2.dat};
               \addlegendentry{$(1, 0; 1, 0, 0)$}
               \addplot[no marks, dgreen, densely dotted] table {fig2-3.dat};
               \addlegendentry{$(2, 0; 2, 0, 0)$}
               \addplot[no marks, mygrey, dash dot] table {fig2-4.dat};
               \addlegendentry{$(0.65, 0.65; 1, 2.1, -1)$}
            \end{axis}
         \end{tikzpicture}
         \caption{%
            Coherence generated and detected, as measured via the trace distance between the two probability distributions of Eq.~\eqref{eq:prob_dists}, maximized over $\rhoS$, for an open-system evolution described by Eq.~\eqref{eq:atomic_ramsey}. \newline
            The legend gives, in units of $\omega$, $(\gamma_{1 1} = \gamma_{2 2}, \gamma_{3 3}; \gamma_{1 2}, \gamma_{1 3}, \gamma_{2 3})$; in particular, all but the gray dash\hyp dotted line represent cases of purely orthogonal noise~\cite{Chaves2013}, since the non-zero rates are associated in Eq.~\eqref{eq:atomic_ramsey} only to Pauli operators in a direction orthogonal to the Hamiltonian $\omega \sigma_3$.
            The total evolution time is fixed to $\omega t_3 = 3$ and the trace distance is plotted as a function of the intermediate time $0 \leq t_2 \leq t_3$.
         }
         \label{fig:atomic_ramsey}
      \end{figure}
      
      Different behaviors of the CGD capability for various parameter choices are illustrated in Figure~\ref{fig:atomic_ramsey}.
      On the one hand, changing the weights of the noise components can result in even qualitatively different features of the coherences generated and detected along the evolution, characterized, for example, by different locations and number of maxima as a function of the intermediate time~$t_2$.
      
      On the other hand, rather different kinds of noise might exhibit a similar behavior.
      In fact, compare the case of pure dephasing, see Figure~\ref{fig:atomic_clock}, and the purely orthogonal noise represented by the solid blue curve in Figure~\ref{fig:atomic_ramsey}.
      The qualitative and even quantitative evolution of the coherences generated and detected is very similar in the two cases.
      This is particularly relevant since it is well known that, if we want to estimate the value of the frequency~$\omega$ via the Ramsey scheme, pure dephasing and orthogonal noise will limit the optimal achievable precision in a radically different way.
      Pure dephasing enforces the shot-noise limit~\cite{Huelga1997,Escher2011,Demkowicz2012,Haase2018}, which is typical of the classical estimation strategies~\cite{Demkowicz2015}.
      Note that this is the case even if error\hyp correction techniques are applied \cite{Sekatski2017,Demkowicz2017,Zhou2018}.
      Orthogonal noise, instead, allows for super\hyp classical precision~\cite{Chaves2013}, which can be even raised to the ultimate Heisenberg limit by means of error correction~\cite{Sekatski2017,Demkowicz2017,Zhou2018}.
      This provides us with an example of how the capability to generate coherences (in the relevant basis) and later convert them to populations has to be understood as a \emph{prerequisite} to perform tasks which rely on the advantage given by the use of quantum features.
      However, CGD in itself does not guarantee that such an advantage over any possible classical counterpart is actually achieved.
       
   \section{Conclusion}
      In this work we have shown that fulfilling a finite number of criteria is necessary and sufficient to ensure that a given quantum dynamics with time\hyp independent generator cannot generate and subsequently detect coherence.
      Importantly, these conditions are given in terms of the generator of the dynamics itself, which makes them even more convenient when one wants to characterize the evolution of a certain open system.
      In the more general case of a time\hyp dependent generator, an uncountably infinite number of conditions arises.
      We have exemplified our results for the case of a GKSL qubit dynamics, providing the defining properties of generators that give rise to coherence non\hyp generating as well as coherence non\hyp activating maps, and applied our findings to analyze the coherence\hyp generating\hyp and\hyp detecting capabilities of the open\hyp system dynamics describing a typical Ramsey protocol.
      
      Our method provides a way of assessing the interconversion between coherence and population which represents a prerequisite for the potential use of coherence as a resource in quantum information technology.
      
   \section*{Acknowledgements}
      The authors thank Andreu Riera and Philipp Strasberg for interesting discussions on various aspects of the present work.
      The authors acknowledge support from Spanish MINECO, project FIS2016-80681-P with the support of AEI/FEDER funds; the Generalitat de Catalunya, project CIRIT 2017-SGR-1127; the ERC Synergy Grant BioQ.
      MGD is supported by a doctoral studies fellowship of the Fundaci\'{o}n ``la Caixa,'' grant LCF/BQ/DE16/11570017.
      MS is supported by the Spanish MINECO, project IJCI-2015-24643. MR acknowledges partial financial support by the Baidu-UAB collaborative project `Learning of Quantum Hidden Markov Models.'
   
   \begingroup
   \RaggedRight
   \bibliographystyle{apsrev4-1}
   \bibliography{ncgd}
   \endgroup
   
   \onecolumngrid
   \appendix
   \newpage
   
   \def\one{\mathbbm1}%
   \makeatletter%
   \def\PP{_{\mathrm{p p}\vphantom j}\@ifstar{^{\vphantom j}}{}}%
   \def\PC{_{\mathrm{p c}\vphantom j}\@ifstar{^{\vphantom j}}{}}%
   \def\CP{_{\mathrm{c p}\vphantom j}\@ifstar{^{\vphantom j}}{}}%
   \def\CC{_{\mathrm{c c}\vphantom j}\@ifstar{^{\vphantom j}}{}}%
   \makeatother%
   \section{Proof of \texorpdfstring{\cref{th:semi}}{\autoref*{th:semi}}}\label{app:proof}
      \subsection{General remarks}
         In the time\hyp independent case, $\map_t = \ee^{t \gen}$, we can write NCGD as
         \begin{equation}
            \Delta \map_t \Delta^\perp \map_\tau \Delta = 0 \quad \forall t, \tau \geq 0\text,
         \end{equation}
         where we defined $\Delta^\perp \coloneqq \one - \Delta$.
         Expanding the exponential, this is equivalent to
         \begin{align}
            \Delta \sum_{n = 0}^\infty \frac{\gen^n t^n}{n!}
            \Delta^\perp \sum_{n' = 0}^\infty \frac{\gen^{n'} \tau^{n'}}{n'!}
            \Delta
            & = 0 \ \forall t, \tau \geq 0 \notag \\
            \Leftrightarrow \Delta \gen^n \Delta^\perp \gen^{n'} \Delta
            & = 0 \ \forall n, n' \in \N_0\text. \label{eqn:tip:equiv}
         \end{align}
         
         Note that for finite\hyp dimensional $\hilb$ the following equivalence holds:
         \begin{equation}
            \mbigl( \gen\PC* \gen\CC^j \gen\CP* = 0 \ \forall j \in \bigl\{ 0, \dots, \ell-1 \bigr\} \mbigr)
            \Leftrightarrow
            \mbigl( \gen\PC* \gen\CC^j \gen\CP* = 0 \ \forall j \in \N_0 \mbigr)\text,
         \end{equation}
         where $\ell= \dim\bound_{\mathrm c}(\hilb) = \dim^2\hilb - \dim\hilb = d^2-d$.
         While ``$\Leftarrow$'' is trivial, by Cayley--Hamilton \cite{Segercrantz1992}, there are coefficients $\{ \alpha_i \}_{i = 0}^\ell\subset \C$, such that
         \begin{equation}
            \sum_{i = 0}^\ell\alpha_{i\vphantom j}^{\vphantom i} \gen\CC^i = 0\text.
         \end{equation}
         Hence any power of $\gen\CC$ greater or equal to $\ell$ can be expressed as a linear combination of powers from $0$ to $\ell-1$.
         All condition formulated in the following that rely on infinite matrix powers are therefore already fulfilled if they hold up to the $(\ell-1)$\textsuperscript{\kern.4pt th} power.
      
      \subsection{Structure lemma}
         \begin{lemma}
            If $\gen\PC* \gen\CC^j \gen\CP* = 0$ for all $j < n$,
            \begin{equation}
               \Delta \gen^n
                 = \mat{\gen\PP^n &
                        \sum_{j = 1}^n \gen\PP^{j -1} \gen\PC* \gen\CC^{n - j}\\
                        0 & 0}
               \quad\text{and}\quad
               \gen^n \Delta
                 = \mat{\gen\PP^n & 0 \\
                        \sum_{j = 1}^n \gen\CC^{j -1} \gen\CP* \gen\PP^{n - j}
                        & 0}\text. \label{eqn:tip:struc}
            \end{equation}
         \end{lemma}
         \begin{proof}
            We will prove the first statement, the second one follows analogously.
            
            For $n = 0$, this is trivially true. Now let $n \mapsto n +1$.
            \begin{align*}
               \Delta \gen^{n +1}
               & = \mat{\gen\PP^n &
                            \sum_{j = 1}^n \gen\PP^{j -1} \gen\PC* \gen\CC^{n - j} \\
                         0 & 0}
                   \mat{\gen\PP & \gen\PC \\ \gen\CP & \gen\CC} \\
               & = \mat{\gen\PP^{n +1} +
                        \bgcol{gray!20!white}{\sum_{j = 1}^n \gen\PP^{j -1} \gen\PC* \gen\CC^{n - j} \gen\CP*} &
                        \gen\PP^n \gen\PC* +
                        \sum_{j = 1}^n \gen\PP^{j -1} \gen\PC* \gen\CC^{n - j +1} \\
                        0 & 0}
               \intertext{By assumption, the highlighted sum is zero. Note that in the second column, the separate term is precisely the one arising for $j = n +1$.}
               \Rightarrow \Delta \gen^{n +1}
               & = \mat{\gen\PP^{n +1} &
                        \sum_{j = 1}^{n +1} \gen\PP^{j -1} \gen\PC* \gen\CC^{n + 1 - j} \\
                        0 & 0}
               \tag*{\qedhere}
            \end{align*}
         \end{proof}
      \clearpage
      
      \subsection{Forwards direction}
         We will prove
         \begin{equation}
            \mbigl( \gen\PC* \gen\CC^j \gen\CP* = 0 \ \forall j \in \N_0 \mbigr)
            \Rightarrow \text{NCGD}\text. \label{eqn:tip:forwards}
         \end{equation}
         
         \begin{proof}
            We directly apply the structure lemma, \cref{eqn:tip:struc}, to \cref{eqn:tip:equiv}.
            The intermediate $\Delta^\perp$ removes the population\hyp to\hyp population entry.
            Hence, we obtain
            \begin{align}
               \Delta \gen^n \Delta^\perp \gen^{n'} \Delta
               & = \mat{0 & \sum_{j = 1}^n \gen\PP^{j -1} \gen\PC* \gen\CC^{n - j} \nonumber\\
                       0 & 0}
                   \mat{0 & 0 \\
                        \sum_{j' = 1}^{n'} \gen\CC^{j' -1} \gen\CP* \gen\PP^{n' - j'} & 0} \nonumber\\
               & = \mat{
                    \sum_{j = 1}^n \sum_{j' = 1}^{n'}
                       \gen\PP^{j -1} \gen\PC* \gen\CC^{n - j}
                       \gen\CC^{j' -1} \gen\CP* \gen\PP^{n' - j'} &
                       0 \\ 0 & 0\label{eq:eq}
                 }\text,
            \end{align}
            and we clearly see another $\gen\PC* \gen\CC^{\cdots} \gen\CP*$ combination, which is zero by the assumption, \cref{eqn:tip:forwards}.
         \end{proof}
         
      \subsection{Backwards direction}
         We will prove
         \begin{equation}
            \text{NCGD}
            \Rightarrow \mbigl( \gen\PC* \gen\CC^j \gen\CP* = 0 \ \forall j \in \N_0 \mbigr)\text. \label{eqn:tip:backwards}
         \end{equation}
         
         \begin{proof}
            By assumption,
            \begin{equation}
               \Delta \gen^n\Delta^\perp \gen \Delta
               = 0 \ \forall n \in \N_0\text,
               \label{eqn:tip:backwards:2}
            \end{equation}
            where we fixed $n' = 1$.
            
            Using induction we will show that if $\gen\PC* \gen\CC^j \gen\CP* = 0$ holds for all $j < n$, then it holds for $j \leq n$.
            Since this can be done for any $n \in \N_0$ and the case for $j=0$, $\gen\PC \gen\CP = 0$, is implied by \cref{eqn:tip:backwards:2} with $n=1$, the statement follows.

            By hypothesis, $\gen\PC* \gen\CC^j \gen\CP* = 0 \ \forall j < n$.
            The structure lemma, \cref{eqn:tip:struc}, therefore applies and hence\footnote{The first identity directly uses the proof of the structure lemma; since we only regard the right column, the lemma also holds for $n +1$.}
            \begin{align*}
               \Delta \gen^{n +1} \Delta^\perp
               & = \mat{0 &
                        \sum_{j = 1}^{n +1} \gen\PP^{j -1} \gen\PC* \gen\CC^{n + 1 - j} \\
                        0 & 0}\text,
               \shortintertext{so that}
               \Delta \gen^{n +1} \Delta^\perp \gen \Delta
               & = \mat{0 &
                        \sum_{j = 1}^{n +1} \gen\PP^{j -1} \gen\PC* \gen\CC^{n + 1 - j} \\
                        0 & 0}
                   \mat{0 & 0 \\
                        \gen\CP & 0} \\
               & = \mat{\sum_{j = 1}^{n +1} \gen\PP^{j -1} \gen\PC* \gen\CC^{n + 1 - j} \gen\CP* & 0 \\
                        0 & 0}\text,
               \intertext{and we can again insert the hypothesis to eliminate all terms except $j = 1$.}
               \Delta \gen^{n +1} \Delta^\perp \gen \Delta
               & = \mat{\gen\PC* \gen\CC^n \gen\CP* & 0 \\ 0 & 0}\text;
            \end{align*}%
            but this, by the assumption of NCGD must be zero, verifying the hypothesis.
         \end{proof}

      \clearpage
      \subsection{Example: NCGD up to third order}\label{app:counterEx}
         Consider a 5-level system where coherence is generated, but not detected, between levels 1 and 2 ($\gen\PC = 0$, $\gen\CP \neq 0$), and where the opposite occurs between levels 4 and 5 ($\gen\CP = 0$, $\gen\PC \neq 0$). At a first instant, coherence is transferred to levels 1 and 3, and, at the next step, to levels 4 and 5, where it is eventually detected. Such a system is described by a rank-3 noise generator with Hamiltonian
         \begin{gather}
            H
            = \frac12 \mat{
                  1 & -1 & 0 & 0&  0\\
                 -1 &  1 & 0 & 0&  0\\
                 0  &  0 & 0 & 0&  0\\
                 0  &  0 & 0 & 1&  1\\
                 0  &  0 & 0 & 1& -1
              }
            \shortintertext{and jump operators}
            J_1
            = \frac{1}{\sqrt{2}} \mat{
                 1   & -\ii & 0 &   0 &    0 \\
                 \ii &   -1 & 0 &   0 &    0 \\
                 0   &    0 & 0 &   0 &    0 \\
                 0   &    0 & 0 &   1 & -\ii \\
                 0   &    0 & 0 & \ii &   -1
              }\!\text, \quad
            J_2
            = \mat{
                 1 & 0 & 0 & 0 & 0 \\
                 0 & 0 & 0 & 0 & 0 \\
                 0 & 1 & 0 & 0 & 0 \\
                 0 & 0 & 0 & 0 & 0 \\
                 0 & 0 & 0 & 0 & 0
              }\!\text{, and} \quad
            J_3
            = \mat{
                 0 & 0 & 0 & 0 & 0\\
                 0 & 0 & 0 & 0 & 0\\
                 0 & 0 & 0 & 0 & 0\\
                 1 & 0 & 0 & 0 & 0\\
                 0 & 0 & 1 & 0 & 0
              }\text.
         \end{gather}
         As expected, such a dynamics is NCGD only up to third order in time, since it can be checked that $\gen\PC \gen\CP = 0$ and $\gen\PC \gen\CC \gen\CP = 0$, but $\gen\PC* \gen\CC^2 \gen\CP* \neq 0$.
         
   \clearpage
   \section{Proof of \texorpdfstring{\cref{ncgd_eq}}{\autoref*{ncgd_eq}}}\label{app:eq}
      \begin{proof}
         ``$\Rightarrow$'':\\
            We first show that if NCGD holds, then
            \begin{align}   \label{NCGD_eq_propagators}
               \map\PC(t_n, t_{n-1}) \map\CC(t_{n-1}, t_{n-2}) \dotsm
               \map\CC(t_{2}, t_{1}) \map\CP(t_{1}, t_{0}) = 0\ \forall t_n \geq \dotsb \geq t_0 \text.
            \end{align}
            If we assume we can expand $\gen$ at the initial time of propagation,
            \begin{align}
               \gen(t)
               & = \gen(t_0) + \order(t - t_0)\text,
               \shortintertext{we immediately obtain}
               \map(t_1, t_0)
               & = \one + \gen(t_0) (t - t_0) + \order\bigl((t - t_0)^2\bigr)\text.
            \end{align}
            Expanding at the corresponding initial times and using linear independence, we can conclude that Eq.~\eqref{NCGD_eq_propagators} then implies the assertion,
            \begin{align}
               \gen\PC(t_n) \gen\CC(t_{n-1}) \dotsm \gen\CC(t_2) \gen\CP(t_1) = 0  \text.
            \end{align}
            
            The proof of \cref{NCGD_eq_propagators} is by induction: \\
               Our hypothesis is that, under the assumption of NCGD, we can conclude
               \begin{equation}
                  \map\PC(t_k, t_{k -1}) \map\CC(t_{k -1}, t_{k -2}) \dotsm
                  \map\CC(t_{2}, t_{1}) \map\CP(t_{1}, t_{0}) = 0\ \forall t_k \geq \dotsb \geq t_0\ \forall 2 \leq k \leq n -1\text.
               \end{equation}
               The case $n = 2$ is exactly the NCGD statement. \\
               To perform the induction step, note that iterating the NCGD condition and using $\Delta^2 = \Delta$ provides us with
               \begin{gather}
                  \Delta \map(t_n, t_{n-1}) \dotsm \map(t_1, t_0) \Delta
                  = \Delta \map(t_n, t_{n-1}) \Delta \dotsm \Delta \map(t_1, t_0) \Delta \\
                  \Leftrightarrow \Delta \map(t_n, t_{n-1}) \bigl[
                         \map(t_{n -1}, t_{n -2}) \dotsm \map(t_2, t_1) -
                         \map\PP(t_{n -1}, t_{n -2}) \dotsm \map\PP(t_2, t_1)
                      \bigr] \map(t_2, t_1) \Delta
                  = 0
                  \intertext{Now using $\map \equiv \map\PP + \map\PC + \map\CP + \map\CC$ (composition of incompatible domains is defined to be zero) we obtain}
                  \begin{aligned}
                     \Leftrightarrow \Delta \map(t_n, t_{n-1}) \bigl[
                     & \map\PC(t_{n -1}, t_{n -2}) \map\CC(t_{n -2}, t_{n -3}) \dotsm
                       \map\CC(t_{2}, t_1) + {} \\
                     & \map\CC(t_{n -1}, t_{n -2}) \map\CC(t_{n -2}, t_{n -3}) \dotsm
                       \map\CP(t_{2}, t_1) + {} \\
                     & \map\CC(t_{n -1}, t_{n -2}) \map\CC(t_{n -2}, t_{n -3}) \dotsm
                       \map\CC(t_2, t_1) \\
                     \bigr]
                     & \map(t_1, t_0) \Delta
                     = 0
                  \end{aligned}
               \end{gather}
               Here, we used the hypothesis to eliminate all cross\hyp terms that contain terms of the form $\map\PC \map\CC \dotsm \map\CC \map\CP$. The first term, combined with the preceding time evolution $\map(t_1, t_0) \Delta$, and also the second term, combined with the succeeding time evolution $\Delta \map(t_n, t_{n -1})$, also vanish by the same argument. Hence, all that remains is
               \begin{equation}
                  \Delta \map(t_n, t_{n-1})
                  \map\CC(t_{n -1}, t_{n -2}) \map\CC(t_{n -2}, t_{n -3}) \dotsm
                  \map\CC(t_2, t_1) \map(t_1, t_0) \Delta = 0\text,
               \end{equation}
               closing the induction.
         
         ``$\Leftarrow$'':\\
            We look at the evolution given by $\Delta \map(t_1,t_0)\Delta$. Let us divide this evolution into small parts (of size $\d s=(t_1-t_0)/n$) $\Delta \map(t_1,t_1-\d s) \cdots \map(t_0+\d s,t_0) \Delta$. By linearly approximating each (which becomes exact in the limit $n\rightarrow\infty$), we get 
            \begin{flalign}
               \MoveEqLeft
               \begin{aligned}[t]
                  & \Delta\bigl[
                    \one + \gen\PP(t_1 - \d s) + \gen\PC(t_1 - \d s) +
                    \gen\CP(t_1 - \d s) + \gen\CC(t_1 - \d s) \bigr] \\
                  & \dotsm
                    \bigl[ \one + \gen\PP(t_0 + \d s) + \gen\PC(t_0 + \d s) +
                           \gen\CP(t_0 + \d s) + \gen\CC(t_0 + \d s) \bigr] \Delta
               \end{aligned} \notag \\*
               = {}
               & \bigl[ \Delta + \gen\PP(t_1 - \d s) + \gen\PC(t_1 - \d s) \bigr]
                 \bigl[ \one + \gen\PP(t_1 - 2 \d s) + \gen\PC(t_1 - 2 \d s) +
                        \gen\CP(t_1 - 2 \d s) + \gen\CC(t_1 - 2 \d s) \bigr] \notag\\*
               & \dotsm
                 \bigl[ \one + \gen\PP(t_0 + 2 \d s) + \gen\PC(t_0 + 2 \d s) +
                        \gen\CP(t_0 + 2 \d s) + \gen\CC(t_0 + 2 \d s) \bigr]
                 \bigl[ \Delta + \gen\PP(t_0 + \d s) + \gen\CP(t_0 + \d s) \bigr] \\
               = {}
               & \bigl[ \Delta + \gen\PP(t_1 - \d s) \bigr] \dotsm
                 \bigl[ \Delta + \gen\PP(t_0 + \d s) \bigr]\text,
            \end{flalign}
            where we used that by assumption any term of the form $\gen\PC(r_j)  \gen\CC(r_{j-1}) \dotsm \gen\CC(r_{k+1}) \gen\CP(r_k)$ is zero.
            This equality means that we get the same total population transfer whether we dephase at every point in time or we do not. Inserting a dephasing at the wanted time and using the equality backwards we get NCGD.
      \end{proof}
\end{document}